\documentclass[10pt, twocolumn]{article}

\usepackage{amsmath,amssymb,amsthm}    
\usepackage{graphics,epsfig,subfigure}    
\usepackage{verbatim}   
\usepackage{color}      

\usepackage{subfigure}  
\usepackage{hyperref}   
\usepackage{epstopdf}
\usepackage{graphicx}
\usepackage{wasysym}
\usepackage{wrapfig}
\usepackage{sidecap}
\usepackage{float}
\usepackage{subfigure}
\usepackage{enumerate}
\usepackage{graphicx}
\usepackage{algorithmic}
\usepackage{algorithm}

\usepackage{authblk}

\newcommand{\bv}[1]{{\boldsymbol{#1} }}

\allowdisplaybreaks
\newtheorem{lemma}{\textbf{Lemma}}
\begin{document}

\title{A Multi-Layer Market for Vehicle-to-Grid Energy Trading in the Smart Grid}
\author[1]{Albert Y.S. Lam
\thanks{Email: \href{mailto:ayslam@eecs.berkeley.edu}{ayslam@eecs.berkeley.edu}
To whom correspondence should be addressed.}}
\author[1]{Longbo Huang
\thanks{Email: \href{mailto:huang@eecs.berkeley.edu}{huang@eecs.berkeley.edu}}}
\author[1]{Alonso Silva
\thanks{Email: \href{mailto:asilva@eecs.berkeley.edu}{asilva@eecs.berkeley.edu}
To whom correspondence should be addressed.}}
\author[2]{Walid Saad
\thanks{Email: \href{mailto:walid@miami.edu}{walid@miami.edu}}}
\affil[1]{University of California, Berkeley}
\affil[2]{University of Miami}
\date{}

\maketitle

\begin{abstract}
In this work, we propose a multi-layer market for vehicle-to-grid energy trading. In the macro layer, we consider a double auction mechanism, under which the utility company act as an auctioneer and energy buyers and sellers interact. This double auction mechanism is strategy-proof and converges asymptotically. In the micro layer, the aggregators, which are the sellers in the macro layer, are paid with commissions to sell the energy of plug-in hybrid electric vehicles (PHEVs) and to maximize their utilities. We analyze the interaction between the macro and micro layers and study some simplified cases.
Depending on the elasticity of supply and demand, the utility is analyzed under different scenarios.  Simulation results show that our approach can significantly increase the utility of PHEVs. 
\end{abstract}

\section{Introduction}
The rising oil prices combined with the ongoing trend for developing environmental-friendly technological solutions, implies that electrically-operated vehicles will lie at the heart of future transportation systems. In particular, it is envisioned that plug-in hybrid electric vehicles~(PHEVs), which are essentially electric vehicles equipped with storage devices, will constitute one key component towards realizing the vision of green, environment-friendly transportation networks. For instance, it is forecast that up to 2.7~million electric vehicles will be put on the road in the United States, by 2020~\cite{evnum}.

The presence of energy storage devices implies that PHEVs can not only serve as a green means of transportation, but also, if properly configured, they can function as a ``moving'' energy reservoir that can store and, possibly, supply power back to the power grid. While current PHEV deployment are mostly concerned with grid-to-vehicle interactions, enabling two-way vehicle-to-grid~(V2G) interactions between the grid and PHEVs, has recently started to receive considerable attention both in research and standardization agencies~\cite{PHEV07,UC01,UC04,PHEV09,PHEV10} and is expected to lie at the heart of the emerging smart grid system.

Enabling V2G interactions has several advantages such as serving as backup power sources during outages or supplying ancillary services back to the grid for regulation services. However, in order to fully reap the benefits of V2G systems, several key challenges must be addressed at different level such as control, communications, charging behavior, implementation, and market mechanisms. In \cite{v2g-islanded10}, the authors propose a scheme that uses PHEV batteries to absorb the randomness in intermittent wind power generation. The authors in \cite{Hamed} study the use of game theory for providing frequency regulation through V2G operation. Using the PHEVs as storage units is studied and analyzed in~\cite{PHEV-managing11} while communication architectures suitable for V2G systems are discussed in~\cite{v2g-architeccture11}. Further, the authors in~\cite{v2g-ancillary11} considers the problem of optimally providing energy and ancillary services using electric vehicles.

Clearly, most existing work are focused on implementation, communication, and energy transfer in V2G systems. However, the need for energy transfer and exchange from PHEVs to the grid has also an economical aspect that must be addressed. In this respect, the work in~\cite{Walid} sheds a light on this aspect by investigating the price and quantities exchanged if the PHEVs and the grid elements form an energy trade market. Beyond~\cite{Walid}, little work seems to have been focused on the economics of V2G exchanges which are essential for a better understanding on the potential of using V2G in future power grid systems.

The main contribution of this paper is to propose a general framework and algorithm for studying the economics of the market emerging between PHEVs, aggregators, and the smart grid elements. To address this problem, we propose a multi-layered market mechanisms in which the agreggators, the PHEVs, and the grid elements can decide on the quantity and prices at which they wish to trade energy while optimizing the tradeoff between the benefits (e.g., revenues) and costs from this energy exchange. In this proposed market, first, the aggregators and the smart grid elements (e.g., substations) submit their reservation prices and bids so as to agree on a price and energy trading mechanism. These interactions are modeled using a double auction whose result is then fed back into the second layer, which deals with the management of PHEV resources at each aggregator. In this layer, within each PHEV group, the aggregator negotiates with the PHEVs to settle for an agreement on the resource usage. In particular, the aggregator will announce its energy buying price to the PHEVs, and each PHEV determines how much energy it is willing to supply to the aggregator. The outcome of this layer is directly linked to the previous market due to the fact that the aggregator has to carefully balance its earnings from the energy market and its payments to the PHEVs within its group. Hence, unlike existing work such as in~\cite{Walid}, the proposed scheme depends, not only on pricing issues, but also on modeling the PHEVs-to-aggregator interactions (from an economical perspective) as well as on providing incentives for the PHEVs to participate in the foreseen market and, hence, it can leverage V2G, for improving the overall smart grid performance. We characterize the equilibria resulting from the proposed multi-layered market and we show their existence. Further, using linear approximation techniques we provide a large-system analysis on the economics of V2G energy trading. Using simulations, we assess the properties and performance of the energy trading mechanisms resulting from proposed scheme and we show that our approach can significantly increase the utility of PHEVs.

%
%
%
%

The paper is organized as follows: In Section \ref{section:model}, we present the proposed system model. In Section \ref{section:linear-approx}, we analyze the system for different PHEV energy supply costs. Simulation results are discussed and analyzed in Section \ref{section:simulation} while conclusions are drawn in Section~\ref{section:conclusion}.


\section{The Model and the Market Mechanism} \label{section:model}

 \begin{figure}[t!]
 \centering
 \includegraphics[height = 2.4in, width=2.2in]{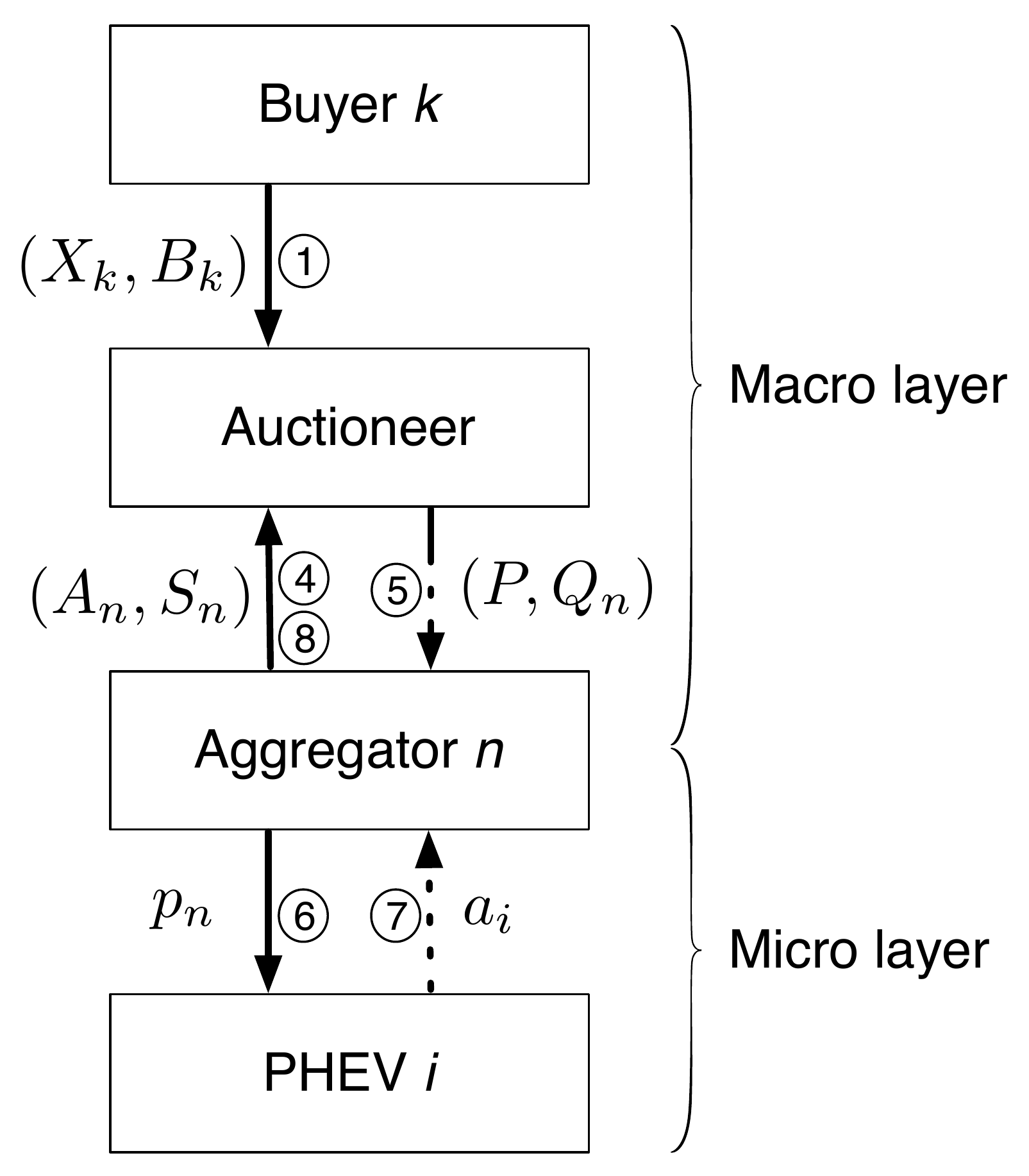}
 \caption{The two-layer market model. The numbers in the figure correspond to the step numbers in Algorithm~\ref{market_mechanism}.}
 \label{fig:market}
 \end{figure}
Consider a smart grid system consisting of $K$ grid elements (e.g., substations) with $\mathcal{K}=\{1,\ldots,K\}$ denoting the set of all such elements. In this grid, $N$ electric vehicles aggregators are deployed. We let $\mathcal{N}=\{1,\ldots,N\}$ denote the set of all aggregators. Each aggregator~\mbox{$n\in \mathcal{N}$} manages  a group of PHEVs denoted by $\mathcal{I}_n=\{1,\ldots,I_n\}$. In this model, we are particularly interested in smart grid elements that are unable to meet their demand and, hence, need to buy energy from alternative sources such as the PHEV aggregators. Hence, hereinafter, all grid elements are referred to as buyers while the aggregators are referred to as sellers.

The energy exchange process between the aggregators and the grid elements is modeled using a two-layered market model as shown in Fig.~\ref{fig:market} with the utility company's control center acting as a middleman that handles the prospective energy trading mechanisms. On the one hand, at the first layer, referred to as the \emph{macro layer}, the aggregators and the grid elements interact so as to trade energy. At this layer, the buyers wish to optimize their performance and meet their demand by buying energy from the PHEVs through the aggregators while the aggregators wish to strategically choose their price and quantity to trade so as to optimize their revenues. On the other hand, at the second layer, referred to as the \emph{micro layer}, the aggregators must interact with the PHEVs so as to optimize the energy resources and provide incentives for the PHEVs to actually participate in the trade. Clearly, the outcomes of these two layers are coupled and, thus, any market solution must take into account this inter-layer dependence. Below, we discuss and analyze, in details, the market operation at each layer.

\subsection{Macro Layer}

An auction is a mechanism for buying or selling goods or services by taking bids and then selling them to the highest bidders.
The main approach to study this mechanism is game theory by considering the bidders as the players of this game,
and their bids as their strategies. For more details on auction theory see~\cite{krishna}.
For the completeness of this work, here we give an overview
of the double auction mechanism~\cite{Huang2002,doubleauction2,Walid}.
In this layer, each potential seller or aggregator~$n\in\mathcal{N}$ 
sends the quantity of energy~$A_n$ that it intends to supply 
and its reservation price~$S_n$ to the auctioneer. 
The reservation price sent by the potential sellers corresponds to the minimum price at which the seller is willing to sell its offered amount of energy. Each buyer~$k\in\mathcal{K}$ proposes a bid $B_k$ and  the quantity it requests, denoted by $X_k$, to the auctioneer. Here, we are mainly focused on the interactions between buyers and sellers in a certain window of time during which the bids and reservation prices do not vary. This can correspond to an energy trading market in which decisions are based on medium or long-term energy needs such as in a day-ahead market. In each round, each aggregator $n$ decides its own $A_n$ with the fixed $S_n$. After receiving all $A_n$'s,
 the auctioneer determines the price~$P(\textbf{A})$ of the energy, where $\textbf{A}=(A_n,1\leq n \leq N)$, and $Q_n(\textbf{A})$,
which corresponds to the total quantity sold by aggregator $n$,
$\forall n\in \mathcal{N}$, by a double auction.

In this work, we adopt a double auction mechanism based on~\cite{Huang2002, Walid, doubleauction1, doubleauction2}.
In this auction:
\begin{itemize}
\item we order the sellers in an
increasing order of their reservation price.
W.l.o.g. we consider
\begin{equation}\label{increasing_supplies}
S_1<S_2<\ldots<S_N.
\end{equation}
\item we order the buyers
in a decreasing order 
of their reservation bids. 
W.l.o.g. we consider
\begin{equation}\label{decreasing_bids}
B_1>B_2>\ldots>B_K.
\end{equation}
\item if two sellers (respectively, buyers)
have equal reservation prices (bids),
they are aggregated into one
single ``virtual'' seller (or buyer).
\item we generate the supply curve
(selling reservation price $S_n$
versus the amount of energy
put out for sale~$A_n$) 
\item we generate the demand curve 
(offered bids~$B_k$ versus quantity needed~$X_k$). 
\item we find an intersection point.
\end{itemize}
%

%
This intersection is at the level
of a certain seller~$L$ and buyer~$M$,
such that~$B_M\ge S_L$ and $B_{M+1}<S_{L+1}$.
We find seller~$L$ and buyer~$M$ and
the double auction dictates that the first $L-1$
and $M-1$ buyers will participate in the energy trading.
We don't consider seller~$L$
and buyer~$M$ in this trade since it's a necessary condition
for matching the total supply and demand
while maintaining a strategy-proof
mechanism~\cite{doubleauction2}.

Thus all sellers with index~$n<L$
and all buyers with index~$k<M$
become the participants
in the double auction so as to clear
the market.
In consequence, the trading prices
for the sellers and the buyers can be
chosen within any point in the range
$[S_L,B_M]$ \cite{doubleauction1}.
In this work, for any
strategy choice~${\bf A}$ (or $\textbf{A}_n, \forall n\in \mathcal{N}$)
by the sellers (i.e. aggregators), given seller~$L$
and buyer~$M$
at the intersection,
we consider all sellers~$i<L$
and buyers~$k<M$
trade at a price
$P({\bf A})$, given by
\begin{equation}
P({\bf A})=\frac{S_L(\bv{A})+B_M(\bv{A})}{2},
\end{equation}
where the dependence on ${\bf A}$
is due to the fact that
each~${\bf A}$ can give different interaction points with the demand and supply curves, and thus we can have
 different~$L$
and~$M$.

At the end of the auction,
numerous criteria can be used
for determining
the amount of traded energy
between each one of the
$L-1$ sellers and~$M-1$ buyers.
We adopt the approach
of~\cite{doubleauction2}
in which the volume
is divided in such a way
as to ensure
a strategy-proof auction.
Using the method in~\cite{doubleauction2}
the total quantity $Q_n({\bf A})$
sold by any~PHEV group~$n$,
for a given choice~${\bf A}$
is:
\begin{equation}
Q_n({\bf A})=
\left\{
\begin{array}{ll}
A_n & \text{if } \sum_{k=1}^{M-1} X_k\ge \sum_{j=1}^{n} A_j,\\
(A_n-\Psi)^+ & \text{if } n=L-1,\\
0 & \text{if } n> L-1,
\end{array}
\right.
\end{equation}
where $(x)^+=\max\{0,x\}$ and $\Psi$ is the oversupply, i.e., $\Psi=\sum_{j=1}^{L-1} A_j - \sum_{k=1}^{M-1} X_k$.

\subsection{Micro Layer}
In this layer,
each aggregator announces a price $p_n$ to its PHEVs by
\begin{equation}
p_n = \gamma_n P(\textbf{A}) ,\forall n\in \mathcal{N}, \label{eq:bid-price}
\end{equation}
where $0<\gamma_n<1$ is the commission rate. The price difference, $P(\textbf{A})-p_n$, is the commission (i.e., cost of management) earned by aggregator $n$ from its managed PHEVs. Each aggregator can make a profit in this way
and thus it has an incentive to help its PHEVS participate in the market.
This provides an incentive for an aggregator to maximize its PHEVs' profits.
By utility maximization, each PHEV $i\in\mathcal{I}_n$  determines its available supply~\mbox{$a_i\in[0,a_i^{\max}]$} according to its utility 
where   $a_i^{\max}$ characterizes the amount of energy the owner of PHEV $i$ can afford to sell, i.e., after reserving enough for its own use.
For the same type of vehicle, a PHEV which needs  more reservation for its proper operation  has a lower $a_i^{\max}$.
 We have
\begin{equation}
A_n = \sum_{i=1}^{I_n}{a_i}, \forall n\in \mathcal{N}. \label{aggreagated_supply}
\end{equation}
Let $\textbf{a}_n=(a_i,1\leq i\leq I_n)$ be the strategy of the PHEVs in group $i$.
Aggregator $n$ determines  the actual quantity $q_i(\textbf{a}_n,Q_n(\textbf{A}))$ allocated to PHEV $i$ in group $n$ proportional to $a_i$ by
\begin{equation}
q_i(\textbf{a}_n,Q_n(\textbf{A})) = \frac{a_i}{\sum_{i\in \mathcal{I}_n}{a_i}} \times Q_n(\textbf{A}),
\end{equation}
and thus
\begin{equation}
Q_n(\textbf{A}) = \sum_{i=1}^{I_n}{q_i(\textbf{a}_n,Q_n(\textbf{A}))}, \forall n\in \mathcal{N}.
\end{equation}



Consider that each PHEV $i\in\mathcal{I}_n$ imposes a cost of discharging its battery to supply energy to its aggregator. We denote this cost by $c_i(a)$, where $a$ is the amount of energy \emph{sold} by PHEV $i$. We assume that during each transaction, each PHEV will try to maximize its profit by choosing the amount of energy $a_i$ to be
\begin{align}
a_i &=  \arg\max_{a_i}\{p_n(\textbf{a}_n)-c_i(a_i)\}.  \nonumber
\end{align}




\subsection{Market Mechanism}
The interaction between the macro and micro layers can be seen as a market mechanism or an algorithm, shown in Algorithm~\ref{market_mechanism}, which is used to find the market equilibrium $P^*, A_n^*$. 
It terminates when a pre-determined number of iterations $t^{\max}$ has been reached or the percentage change of the market price is less than a certain threshold $\xi$.

\begin{algorithm}
\caption{Market Mechanism}
\label{market_mechanism}
\begin{algorithmic}[1]
\STATE For all $k$, buyer $k$ submits its bid $B_k$ and its requested amount $X_k$ to the auctioneer
\STATE $t \gets 1$
\REPEAT
\STATE For all $n$, aggregator $n$ submits its reservation price $S_n$ and its proposed supply~$A_n$ to the auctioneer
\STATE The auctioneer implements a double auction and determines the market price $P(t)$ and the allocated quantity $Q_n$ for aggregator $n$ for all $n$
\STATE For each aggregator $n$, the selling price $p_n$ is announced to its PHEVs
\STATE For each PHEV $i$, it determines its proposed selling amount $a_i$ by maximizing its utility according to  $p_n$ and returns $a_i$ back to its aggregator
\STATE Each aggregator $n$ sums up all $a_i$'s from its PHEVs
\STATE $t \gets t+1$
\UNTIL{$t> t^{\max}$ or $|\frac{P^{(t)}-P^{(t-1)}}{P^{(t)}}|<\xi$}
\end{algorithmic}
\end{algorithm}
As we can see in simulation later, Algorithm 1 performs very efficiently and converges within a few steps. 


\section{Linear Approximation for Large systems}\label{section:linear-approx}
In a practical smart grid, both the number of PHEVs and smart grid element are expected to be large. In this respect, it is of interest to analyze the impact of the presence of large numbers of buyers and aggregators on the overall proposed market mechanism. Therefore, here, we consider the situations when $N$ and $M$ are sufficiently large such that the corresponding supply and demand curves can be approximated by a linear  function, depending on the distribution of aggregators' reservation price (buyers' bids). We assume that the reservation prices $S_n$ from each one can be ordered nicely into a line with each point being one unit apart. Fig.~\ref{figurecontinuous} shows one example of the supply and demand curves, each of which is approximated by one single  linear function. In the following, since the focus of this paper is on the supply side, we study different types of supply curves while always keeping the demand curve of the same type.


\subsection{Linear Cost Function with Homogeneous PHEVs} \label{subsec:linearhomo}
In this subsection, we consider that the cost function of PHEV $i$ is a linear function
with respect to the amount of energy it provides, i.e., $c_i(a)=\eta_i a$,
within a certain interval~$[0,a_i^{\max}]$. If the amount of energy to be sold~exceeds $a_i^{\max}$,
the cost will become infinite. Note that this case also takes into account the inconvenience cost (e.g.,
when the PHEV does not have enough energy to operate normally).
Then its utility function
is given by 
$u_i(a)=(\gamma_n P-\eta_i) a$.
 The problem PHEV $i$ needs to address is the following:
\begin{equation}
u_i^*=\underset{a_i\in[0,a_i^{\max}]}{\text{maximize}}\, u_i(a_i). \label{utilitymax}
\end{equation}
We note that, a PHEV can always decide not to participate in the market in which case it maximum utility will be~$u_i^*=0$. As a result, the maximum utility in (\ref{utilitymax}) is always nonnegative. In this case, the solution $a_i^*$ of (\ref{utilitymax}) is~
\begin{eqnarray}
a_i^*=\left\{\begin{array}{cl}
a_i^{\max} & \text{if}\,\,\,\eta_i\geq p_n,\\
0 & \text{else}.
\end{array}\right.
\end{eqnarray}
%

We can then subdivide the set of PHEVs~$\mathcal{I}_n$ in two disjoint sets~$\mathcal{I}_n^{(1)}$ and $\mathcal{I}_n^{(2)}$,
such that $\mathcal{I}_n^{(1)}\cup \mathcal{I}_n^{(2)}=\mathcal{I}_n$, where
\mbox{$a_i^*=0, \forall i\in \mathcal{I}_n^{(1)}$}, and  $a_i^*=a_i^{\max}, \forall i\in\mathcal{I}_n^{(2)}$.
Thus, we obtain
\begin{equation}
A_n=\sum_{i\in \mathcal{I}_n} a_i=\sum_{i\in \mathcal{I}_n^{(1)}} a_i+\sum_{i\in \mathcal{I}_n^{(2)}} a_i.
\end{equation}
Since for all $i\in \mathcal{I}_n^{(1)}$, $a_i^*=0$, then $\sum_{i\in \mathcal{I}_n^{(1)}} a_i=0$,
and in consequence
\begin{equation}\label{suma}
A_n=\sum_{i\in \mathcal{I}_n^{(2)}} a_i=\sum_{i\in \mathcal{I}_n^{(2)}} a_i^{\max}. \end{equation}

\subsubsection{Linear Approximation}

 \begin{figure}[!t]
 \centering
 \vspace{-.2in}
 \subfigure[Real situation]{\includegraphics[width=1.7in]{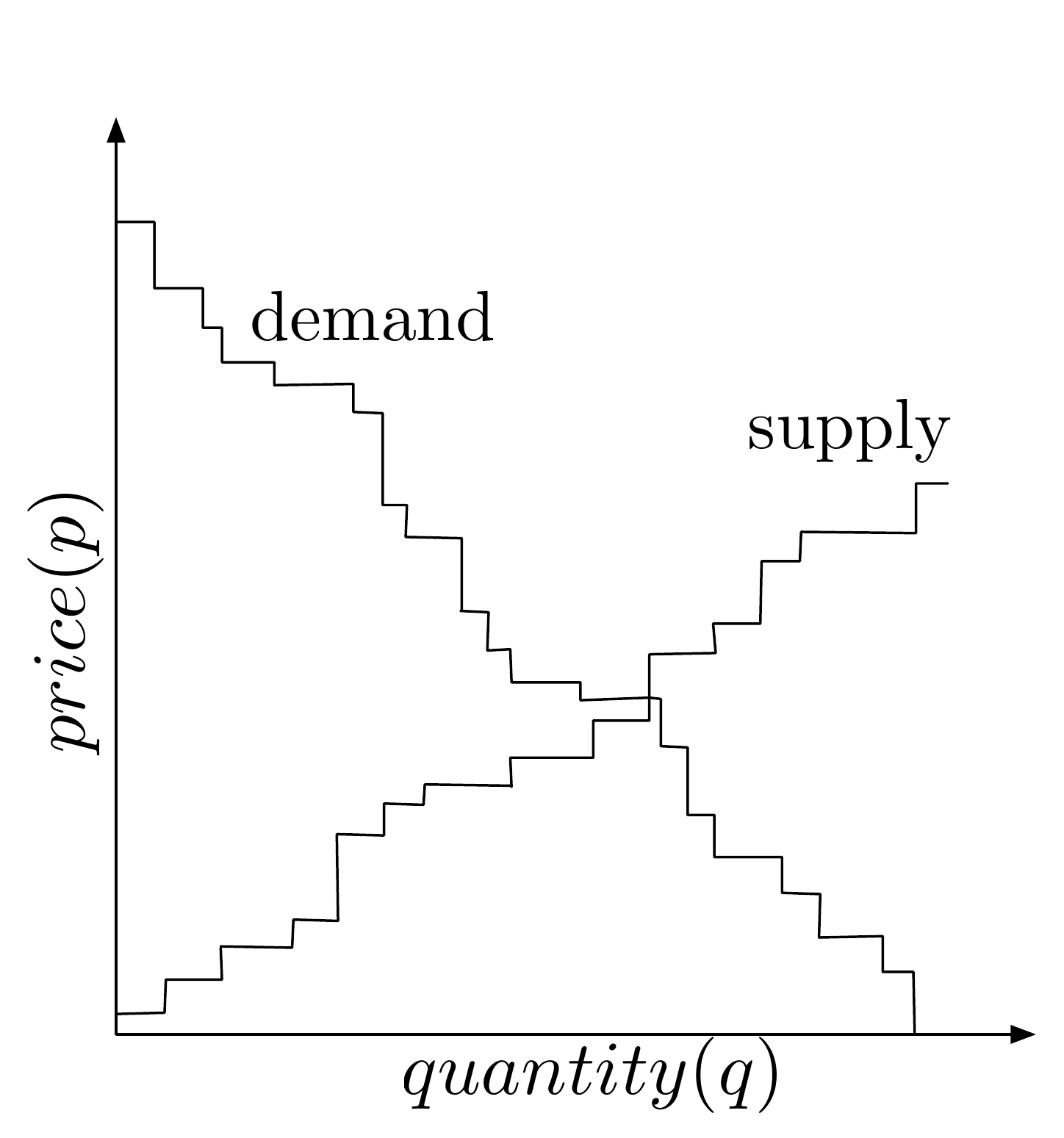}%
 \label{fig:heterogenous}}
 \hfil
 \centering
 \subfigure[Linear approximation]{\includegraphics[width=1.7in]{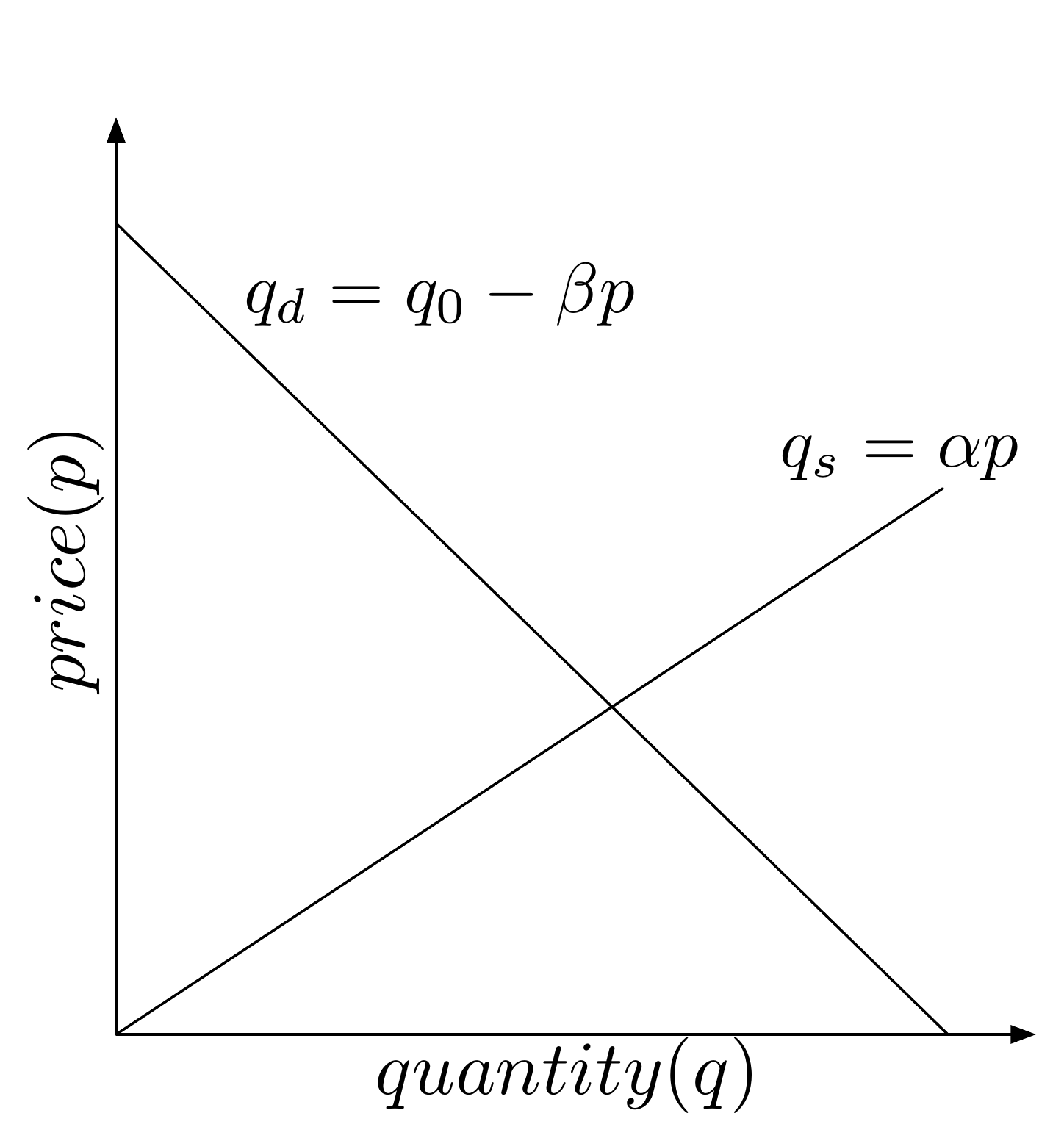}%
 \label{fig:linear}}
 \caption{Supply and demand curves}
 \label{figurecontinuous}
 \end{figure}

Consider the linear case as in Fig.~\ref{fig:linear}.
The PHEVs are price takers and
the supply and demand curves are given by
\begin{eqnarray}
\textrm{Supply}(P,Q)&=&\alpha P,\\
\textrm{Demand}(P,Q)&=&Q_0-\beta P.
\end{eqnarray}
Here~$Q_0$ is the total demand from all the buyers and~$P$ is the price determined from the double auction in the macro layer. Also, $\alpha$ will depend on $\mathcal{I}_n^{(2)}$  for all aggregator $n$. 
In that case, the equilibrium price $P^*$ and the quantity sold $Q^*$ can be determined when the supply meets the demand, i.e.,
\begin{align*}
 \textrm{Supply}(P^*,Q^*)=\textrm{Demand}(P^*,Q^*) \Rightarrow \alpha P^*=Q_0-\beta P^*,
\end{align*}
or equivalently, when
\begin{equation}
P^*=\frac{Q_0}{\alpha+\beta}. \label{eq:equ_price}
\end{equation}
%
%
The utility for each PHEV~$i$ in group $n$ will be thus given by
\begin{enumerate}
\item If $\eta_i\ge\frac{\gamma_n Q_0}{(\alpha+\beta)}$ then $a_i^*=0$ and $u_i^*=0$.
\item If $\eta_i<\frac{\gamma_n Q_0}{(\alpha+\beta)}$ then $a_i^*=a_i^{\max}$ and
\begin{equation}
u_i^*=\left(\gamma_n P^*-\eta_i\right)a_i^* = \left(\frac{\gamma_n Q_0}{(\alpha+\beta)}-\eta_i\right)a_i^{\max}.
\end{equation}
The total utility $U_n$ of aggregator $n$ is given by
\begin{align}
U_n=\sum_{i\in\mathcal{I}_n^{(2)}}{\left(\frac{\gamma_n Q_0}{(\alpha+\beta)}-\eta_i\right)a_i^{\max}}.
\end{align}
\end{enumerate}

From this simple example, we can deduce the following properties:
\begin{enumerate}
\item If the supply slope of a market~$1$ is higher than that of a market~$2$,
i.e., $\alpha_1>\alpha_2$, then the utility gained in market~$1$ is smaller than that in market~$2$.
\item If the demand slope of a market~$1$ is higher than that of a market~$2$,
i.e., $\beta_1>\beta_2$, then the utility gained in market~$1$ is smaller than that in market~$2$.
\item If the total possible demand of a market~$1$ is higher than that of a market~$2$,
i.e., $Q_0^1>Q_0^2$, then the utility gained in market~$1$ is greater than that in market~$2$.
\end{enumerate}

\subsection{Quadratic Cost Function}

We first recall that under the quadratic cost model, each PHEV $i$ will supply the following amount of energy:
\begin{eqnarray}
a_i^*  =  \bigg[\frac{p_n(\textbf{a}_n)-\eta_i}{2\upsilon_i}\bigg]_0^{a_i^{\max}}. 
\end{eqnarray}
In order to make the analysis easier, below we 
assume that $a_i^*\leq a_i^{\max}$ for all $i$. Then,  we have:
\begin{eqnarray}
a_i^*  =  \frac{p_n(\textbf{a}_n)-\eta_i}{2\upsilon_i}, \quad\text{and}\quad A_n^*=C_n p_n(\textbf{a}_n)-d_n,  \label{eq:approx-demand}
\end{eqnarray}
where 
\begin{equation}
C_n= \sum_i\frac{1}{2\upsilon_i}\quad
\textrm{and}\quad
d_n= \sum_i\frac{\eta_i}{2\upsilon_i}.
\end{equation}
In this case, $A_n^*$ is the supply from one aggregator and  $A_n^*$ is the slope of the supply curve, i.e., $\alpha=A_n^*$.
Thus,  we have by (\ref{eq:equ_price}) that:
\begin{eqnarray}
P^* = \frac{Q_0}{A^*_n+\beta}. \label{eq:price-supply}
\end{eqnarray}
Hence, we see that there is an iteration process of the market price, which approximates the actual price iteration:
\begin{itemize}
\item In iteration $t$, with $A_n(t)$ computed, we obtain the optimal price $P^*(t)=\frac{Q_0}{A_n(t)+\beta}$.
\item In iteration $t+1$, we compute the new supply quantity by (\ref{eq:approx-demand}), i.e., $A_n(t+1) = C_n\gamma_n P^*(t)-d_n$.
\end{itemize}
Based on the above observation, we have the following simple lemma characterizing the  necessary conditions under which there exists market equilibriums, i.e., (\ref{eq:approx-demand}) and (\ref{eq:price-supply}) both hold.

\begin{lemma} 
The following conditions are necessary for the above iteration process converges to an  market equilibrium ($P^*, A_n^*$):
\begin{eqnarray}
(\beta+d_n)^2-4(d_n\beta-C_nQ_0)&\geq& 0, \forall n\in \mathcal{N}\label{eq:Acond}\\
(\beta-d_n)^2+4C_nQ_0&\geq& 0, \forall n\in \mathcal{N}.\label{eq:Pcond}\,\,\Diamond
\end{eqnarray}
\end{lemma}
\begin{proof}
In equilibrium, this price must result in a supply that is exactly equal to the resulting $A_n$, i.e.,
\begin{eqnarray}
A_n = \frac{C_nQ_0\gamma_n}{A_n+\beta} - d_n.
\end{eqnarray}
This gives rise to the following condition on $A_n$:
\begin{eqnarray}
 A_n^2 +(\beta+d_n)A_n+\beta d_n-C_n\gamma_nQ_0=0. \label{Q_condition}
\end{eqnarray}
%
%
Note that the above argument also shows an interesting fact that there is an implicit iteration for the market price $P$ as follows:
\begin{eqnarray}
P^*  &=& \frac{Q_0}{ C_n\gamma_nP^*-d_n +\beta}.
\end{eqnarray}
This similarly implies that the fixed point should be:
\begin{eqnarray}
  C_n \gamma_nP^{*2} + (\beta -d_n)P^* -Q_0=0. \label{P_condition}
\end{eqnarray}
Since both \eqref{Q_condition} and \eqref{P_condition} are quadratic functions, we see that the only way there exists a market equilibrium is when  \eqref{eq:Acond} and \eqref{eq:Pcond} hold.
%
%
\end{proof}

In the next section, we will show how the algorithm evolves through simulation and demonstrate the intuition behind the iteration process.


\vspace{-.1in}
\section{Simulation Results}\label{section:simulation}
For simulations, we consider a smart grid network in which a number of aggregators sell their energy surplus to smart grid elements (buyers) through a utility company. The simulation setting  is as follows. Each aggregator manages a certain number of PHEVs, randomly generated in the range of $[500,1000]$. Each PHEV has a maximum battery capacity of 250 miles with power consumption 22kWh per 100 miles \cite{car1,car2} out of which an arbitrary amount of energy, between $30$ and $100$ miles, is reserved for the PHEV's private use. The reservation prices of the aggregators are uniformly selected in $[10,50]$ dollars/MWh while the buyers' bids are randomly chosen from $[15,60]$ dollars/MWh. Each buyer requests energy demand with the amount chosen in $[20,60]$ MWh. The commission rate is set to $\gamma_n = 0.91, \forall n\in \mathcal{N}$. Each PHEV $i$ has random cost function parameters $\eta_i\in[10,50]$ and $\upsilon_i\in [1000,2000]$ for quadratic cost function and $\beta_i =0$ for the linear cost function. The algorithm always starts by setting $a_i=a^{\max}_i,\forall i\in \mathcal{I}_n, \forall n\in \mathcal{N}$.

 \begin{figure}[!t]
 \centering
  \vspace{-.05in}
 \subfigure[Total utility and price for linear cost]{\includegraphics[width=3.4in]{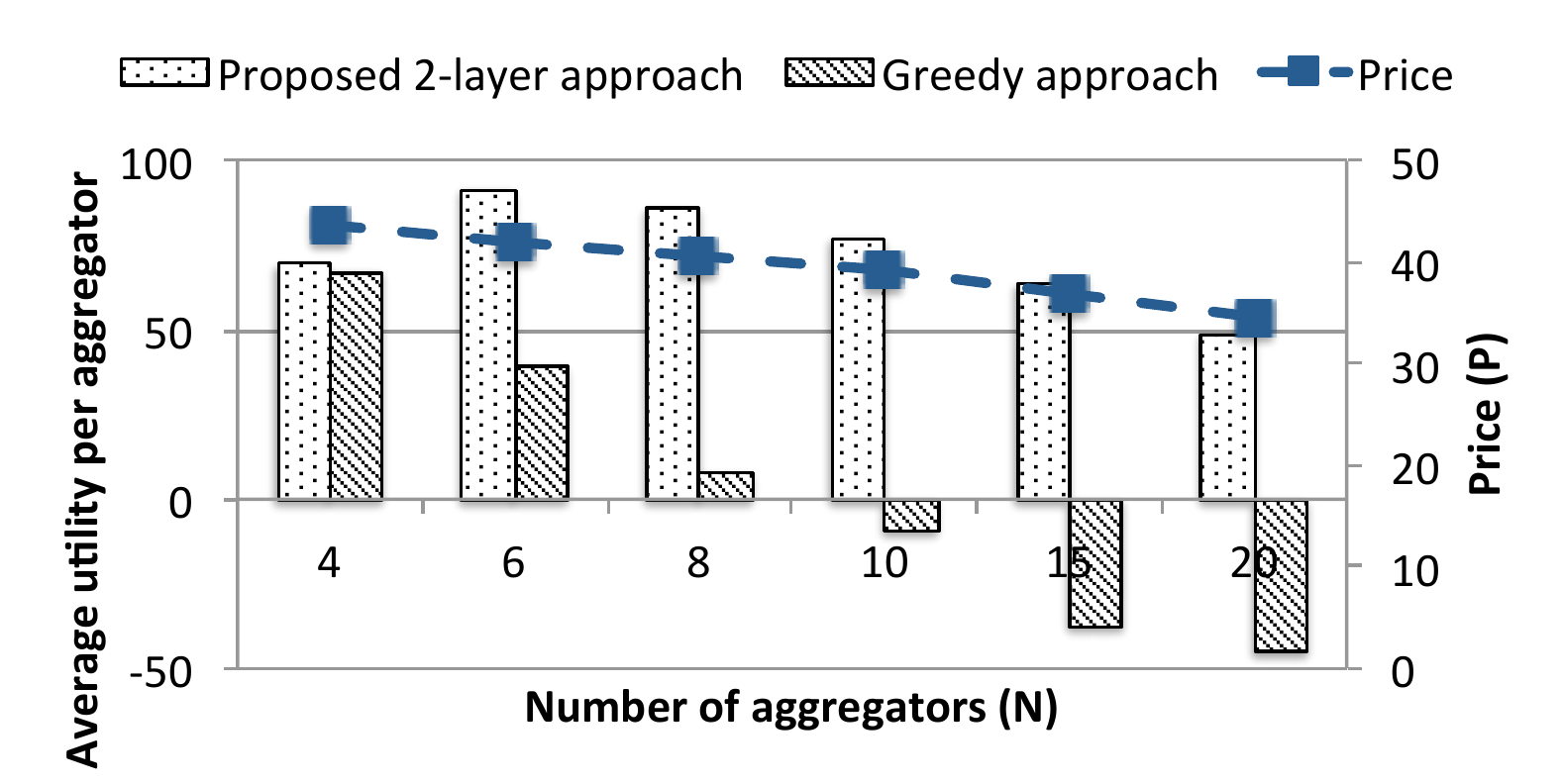}%
 \label{fig:linear_1000}}
 \vspace{-.1in}
 \centering
 \vspace{-.05in}
 \subfigure[Total utility and price for quadratic cost]{\includegraphics[width=3.4in]{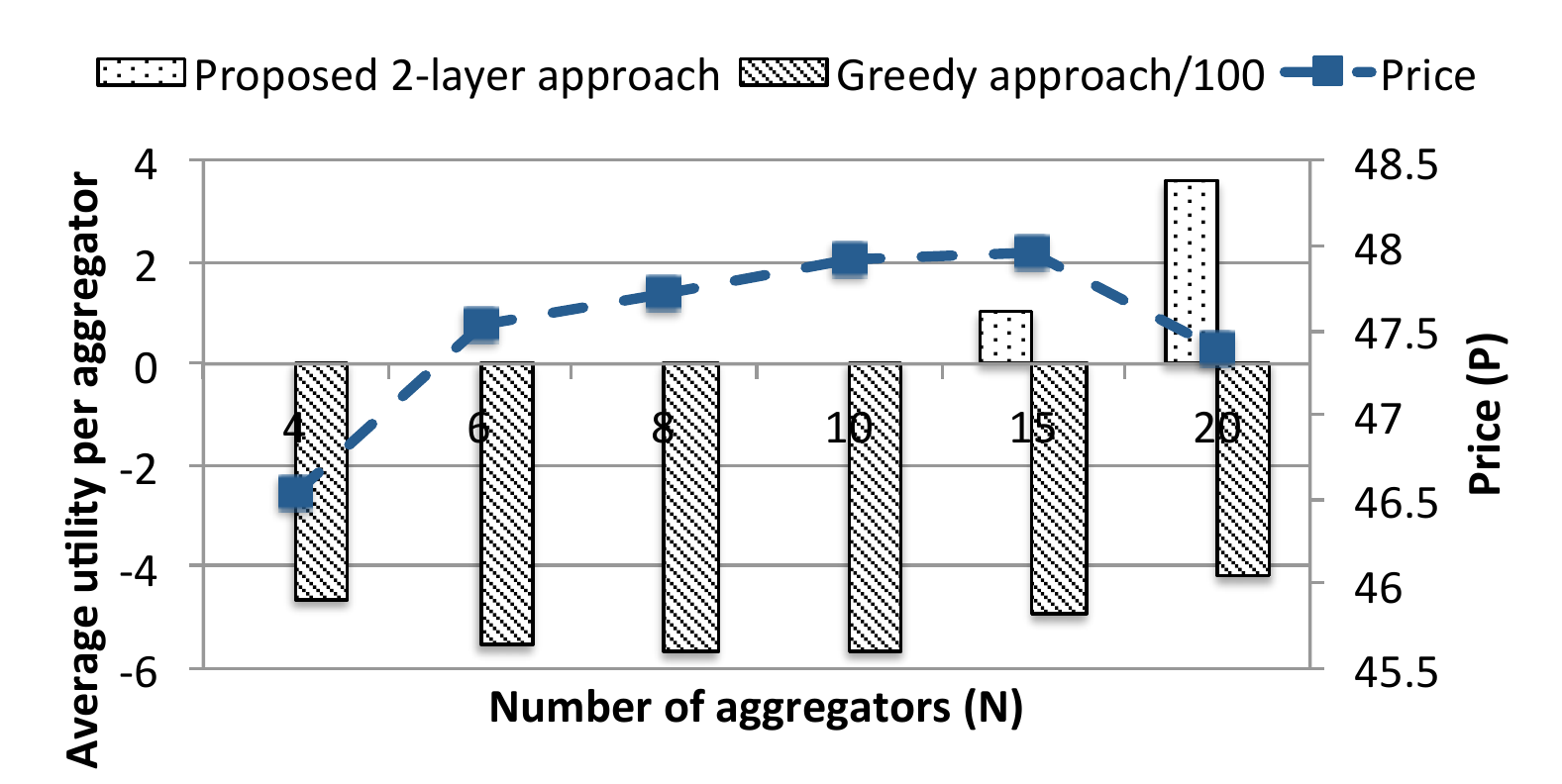}%
 \label{fig:linear_quadratic_1000}}
 \vspace{-.1in}
 \centering
  \vspace{-.05in}
 \subfigure[Iterations required for convergence]{\includegraphics[width=3.4in]{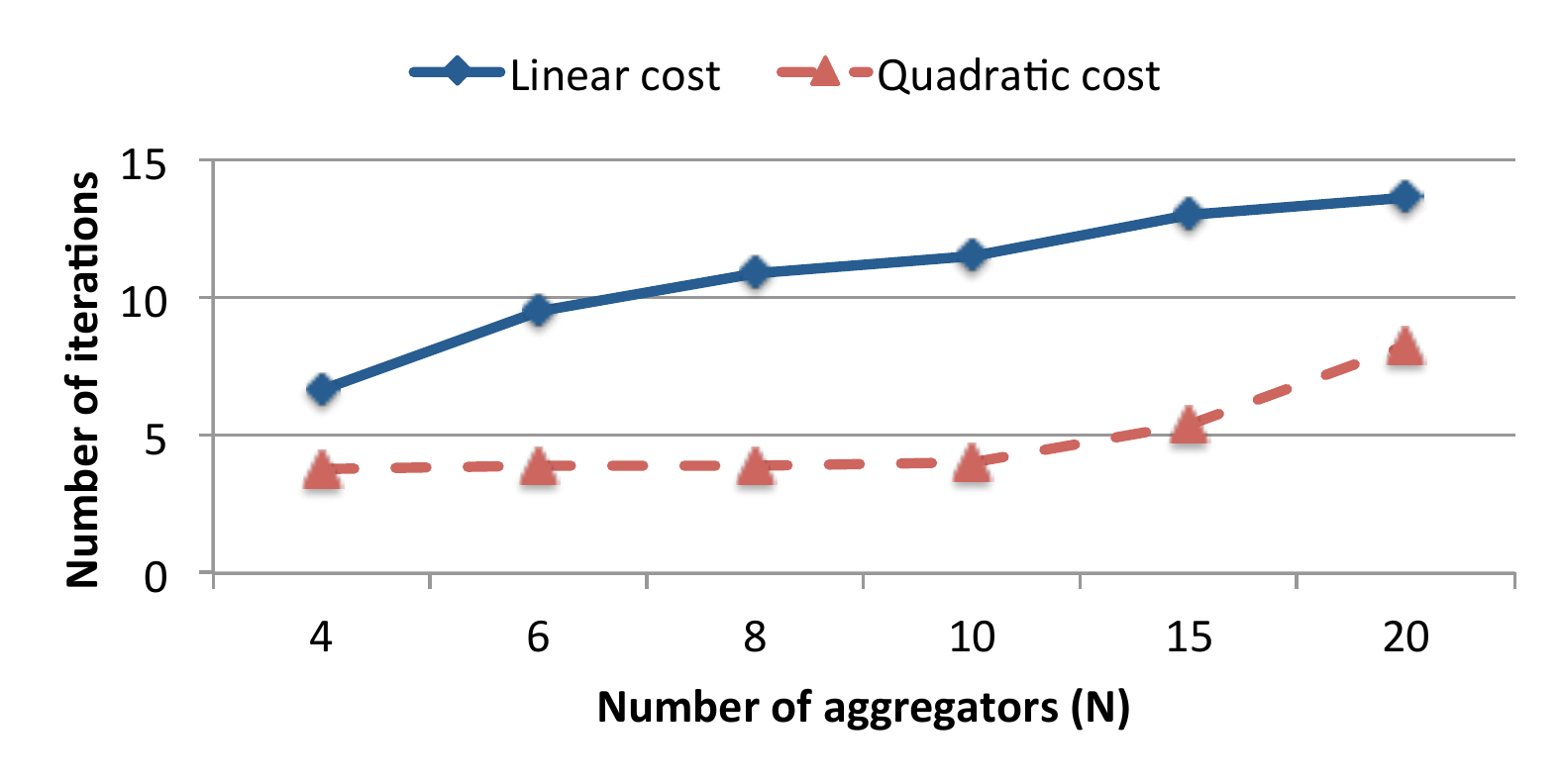}%
 \label{fig:iterations}}
 \caption{Basic simulations for small numbers of buyers and aggregators.}
 \label{fig:sim1000_basic}
     \vspace{-.2in}
 \end{figure}

\subsection{Small Numbers of Buyers and Aggregators}
We simulate the cases with small numbers of buyers and aggregators. We consider 5 buyers ($K=5$) in each case and we compare our two-layer approach with a greedy approach, in which each PHEV always proposes to sell $a_i^{\max}$. Fig. \ref{fig:sim1000_basic} shows the average results of 1000 independent simulation runs for each case. Figs.~\ref{fig:linear_1000} and \ref{fig:linear_quadratic_1000} present the average utility per aggregator corresponding to the linear and quadratic cost functions, respectively. We can see that our approach always gives higher average utility than the greedy one, as shown in Figs.~\ref{fig:linear_1000} and \ref{fig:linear_quadratic_1000}. In the cases of linear cost, we can see that the average utility start by increasing with $N$ but then, it starts to decrease when the number of aggregators reaches that of buyers ($N=6$). This result is due to the fact that, for small $N$, an increase in the number of participating aggregators leads to a larger amount of energy sold which, subsequently, improves the average utility. However, when $N\geq 6$, an increase in the number of aggregators $N$ will yield a decrease in the settled price which leads to a decrease in the utility. In the cases with quadratic cost, the average utility increases with $N$ on the grounds that more aggregators participate in the market resulting in a larger amount of total energy sold. Hence, in general, the equilibrium trading price decreases with $N$ as more aggregators lead to an increased competition, which subsequently imposes a lower market price.

Fig.~\ref{fig:iterations} shows the average number of iterations required to reach an equilibrium. Clearly, as more aggregators participate in the market, the number of iterations till convergence increases. Moreover, Fig.~\ref{fig:iterations} shows that the convergence time is faster in the case with quadratic cost. With a linear cost function, each PHEV $i$ takes either 0 or $a_i^{\max}$ in each iteration, and thus, the algorithm is more like to oscillate more around the equilibrium point before convergence. With a quadratic cost function, due to the concavity of the utility, the algorithm moves toward the equilibrium in a smoother manner which is further corroborated in the subsequent simulations.

\subsection{Large Numbers of Buyers and Aggregators}
Here, we simulate cases in which a large number of buyers ($K=1000$) and aggregators ($N=1000$) are deployed so as to verify the analytical results induced from the linear approximation studied in Section \ref{subsec:linearhomo}. As previously mentioned, when $N$ ($K$) increases, the supply (demand) curve approaches a linear function as all random numbers are generated uniformly.
 \begin{figure}[!t]
 \centering
 \subfigure[Total utility in each iteration]{\includegraphics[width=1.7in]{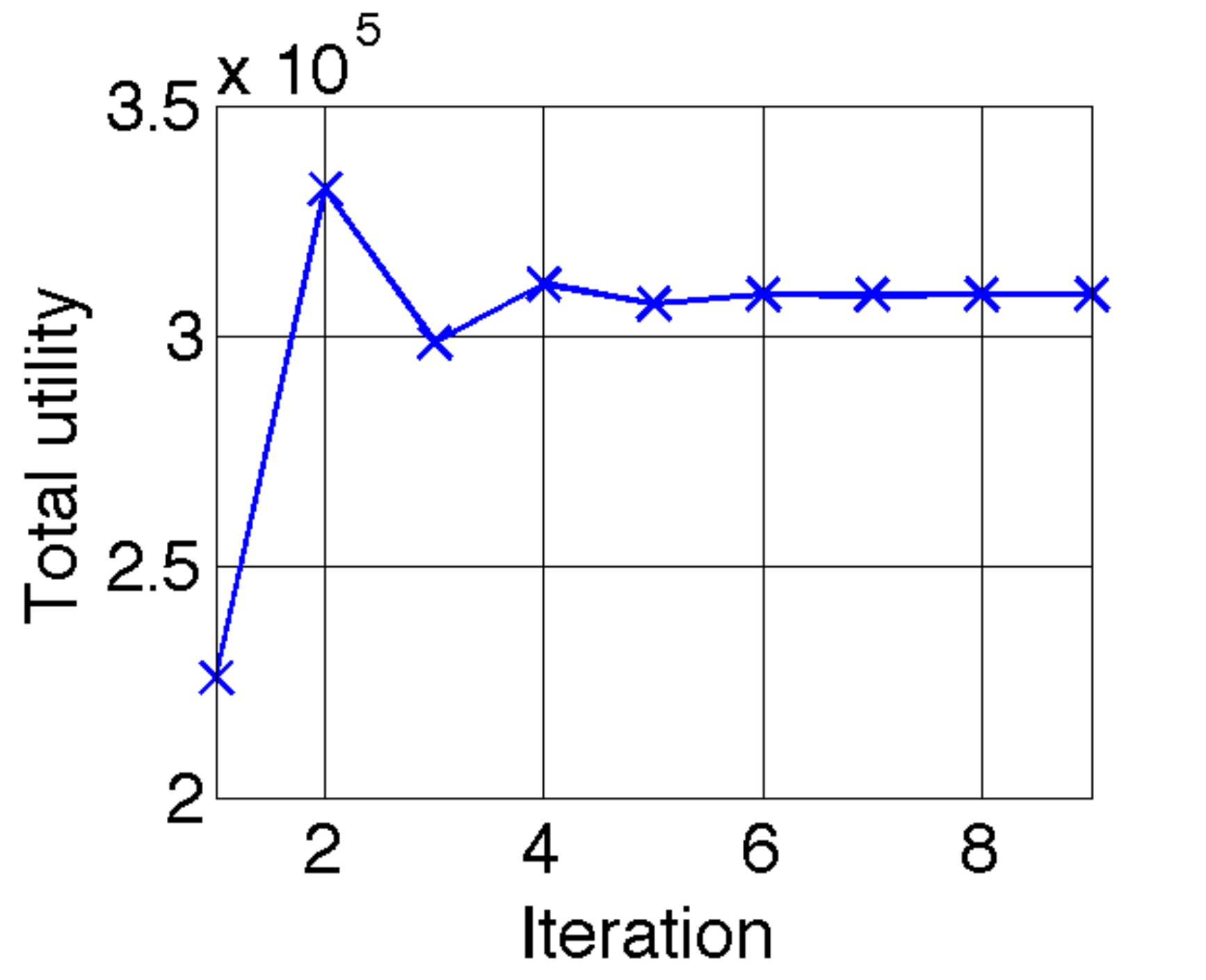}%
 \label{fig:linear_sim_1000_iter}}
 \hfil
 \centering
 \subfigure[Supply and demand curves]{\includegraphics[width=1.7in]{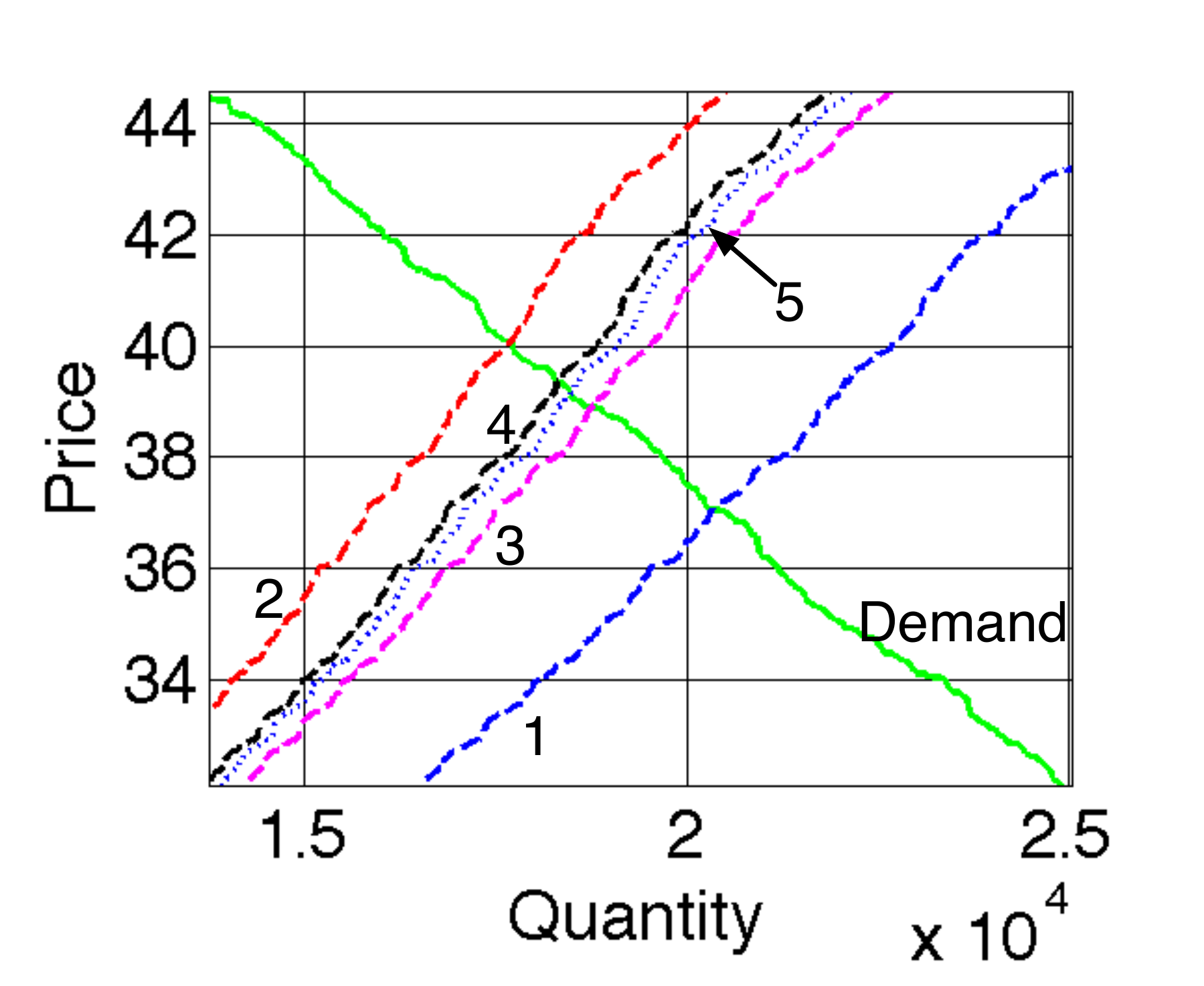}%
 \label{fig:linear_sim_1000_sd}}
 \caption{Linear cost function with 1000 buyers (K) and 1000 aggregators (N)}
 \label{fig:sim1000_linear}
    \vspace{-.2in}
 \end{figure}
 \begin{figure}[!t]
 \centering
 \subfigure[Total utility in each iteration]{\includegraphics[width=1.7in]{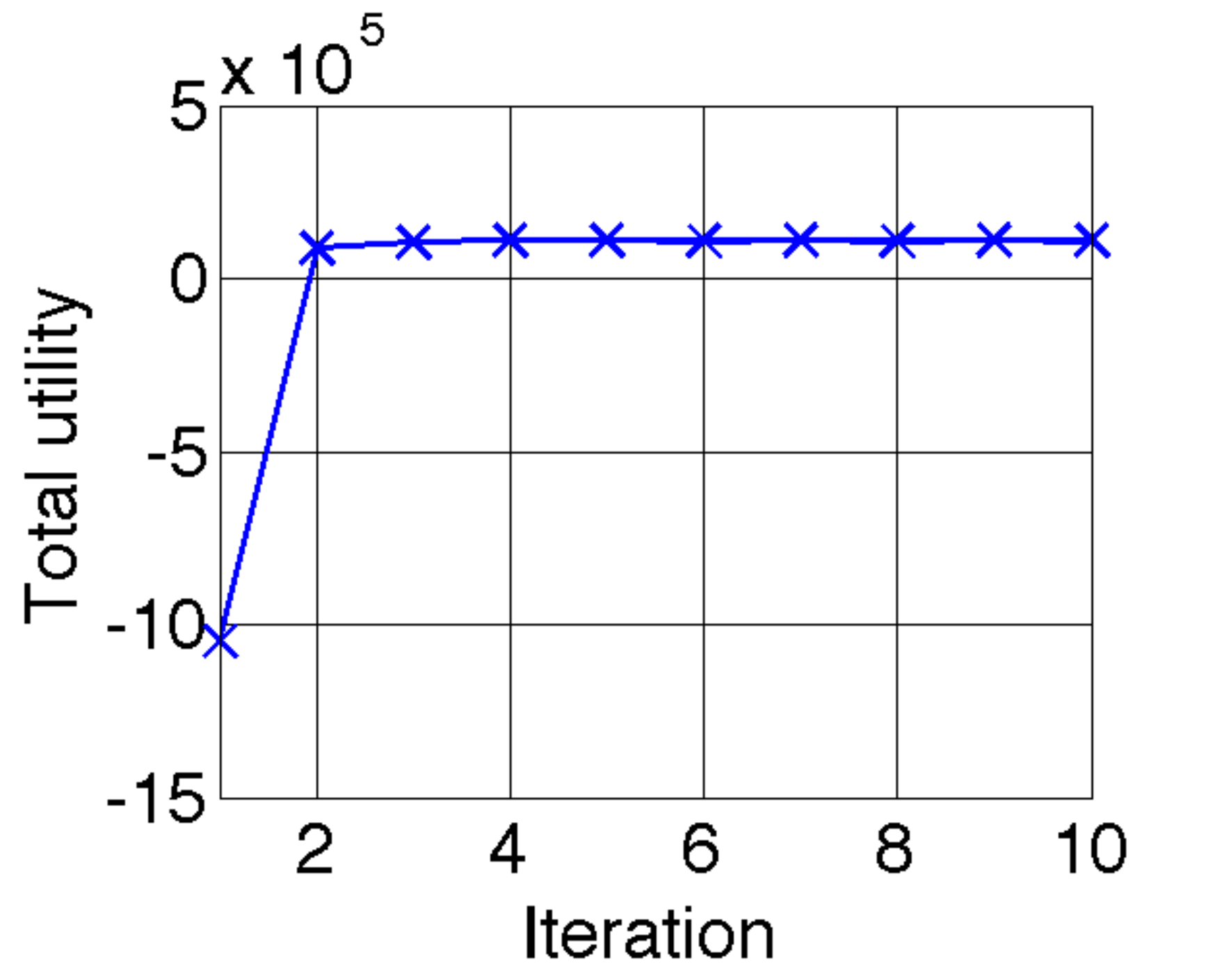}%
 \label{fig:quadratic_sim_1000_iter}}
 \hfil
 \centering
 \subfigure[Supply and demand curves]{\includegraphics[width=1.7in]{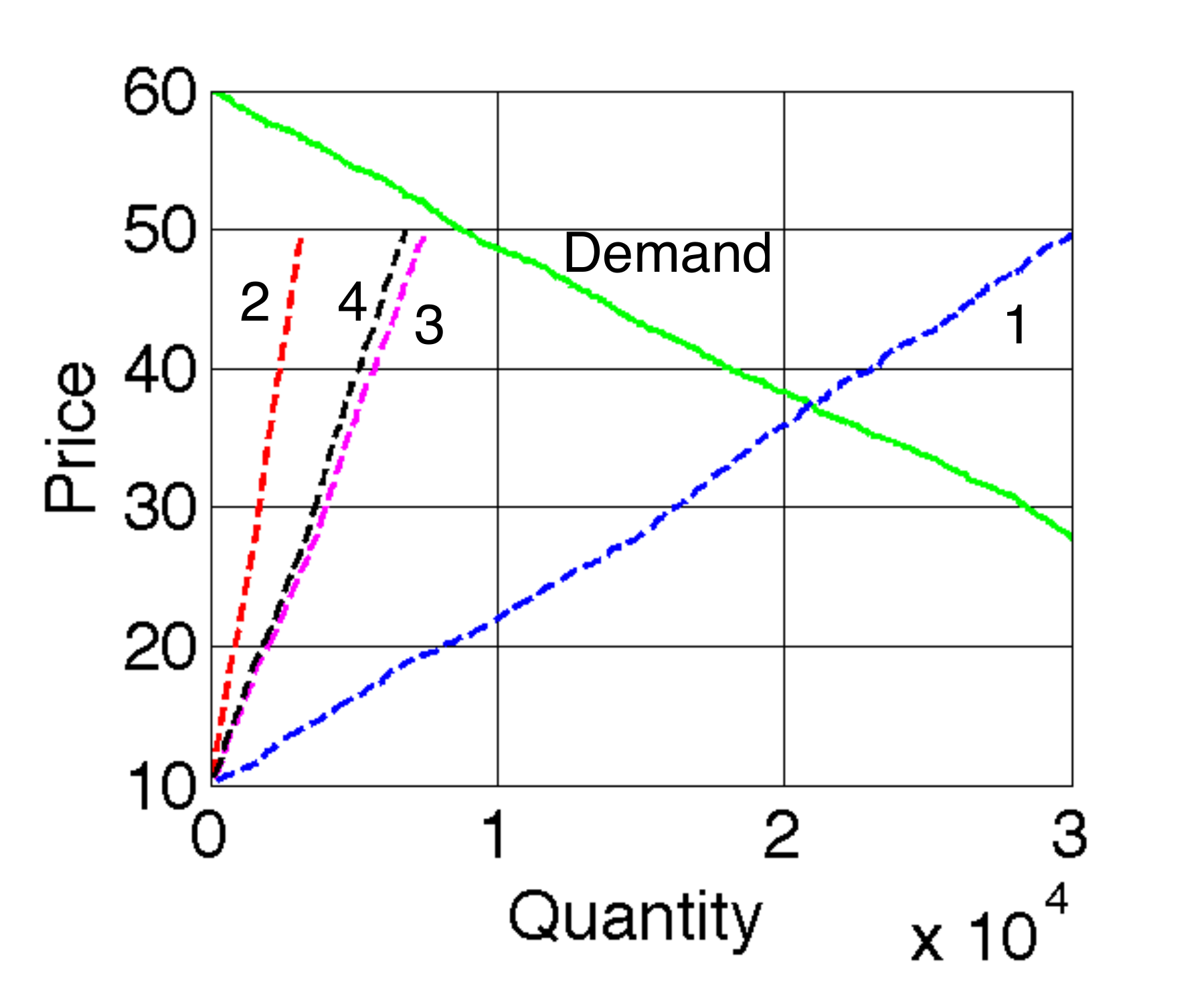}%
 \label{fig:quadratic_sim_1000_sd}}
 \caption{Quadratic cost function with 1000 buyers (K) and 1000 aggregators (N)}
 \label{fig:sim1000_quadratic}
     \vspace{-.2in}
 \end{figure}
We first study the results with a fixed demand curve (i.e. with $\beta$ and $Q_0$ fixed) and they correspond to the situations when the value evolves in a particular simulation.
Fig. \ref{fig:sim1000_linear}\footnote{The numbers in Subfigure (b) correspond to iteration/case numbers in Subfigure (a). For clearer demonstration, In Subfigure (b), only the curves corresponding to the first few iterations/cases are given.} shows the results for linear PHEV cost functions when the algorithm iterates. Fig. \ref{fig:linear_sim_1000_iter} gives the total utility computed from double auction for each iteration while we have the corresponding supply and demand curves in Fig. \ref{fig:linear_sim_1000_sd} ( the part shows intersection points only). We can see that the total utility decreases with the supply slope ($\alpha$). For example, we consider iterations 1 and 2. $\alpha_1$ is larger than $\alpha_2$ and we have the utility gained in iteration 1 is smaller than that in iteration 2.
We also study the quadratic PHEV cost functions and the similar results are shown in Fig.~\ref{fig:sim1000_quadratic}.\footnotemark[\value{footnote}] However, with the quadratic cost function, the algorithm converges faster and smoother.
 \begin{figure}[!t]
 \centering
 \subfigure[Total Utility in each case]{\includegraphics[width=1.7in]{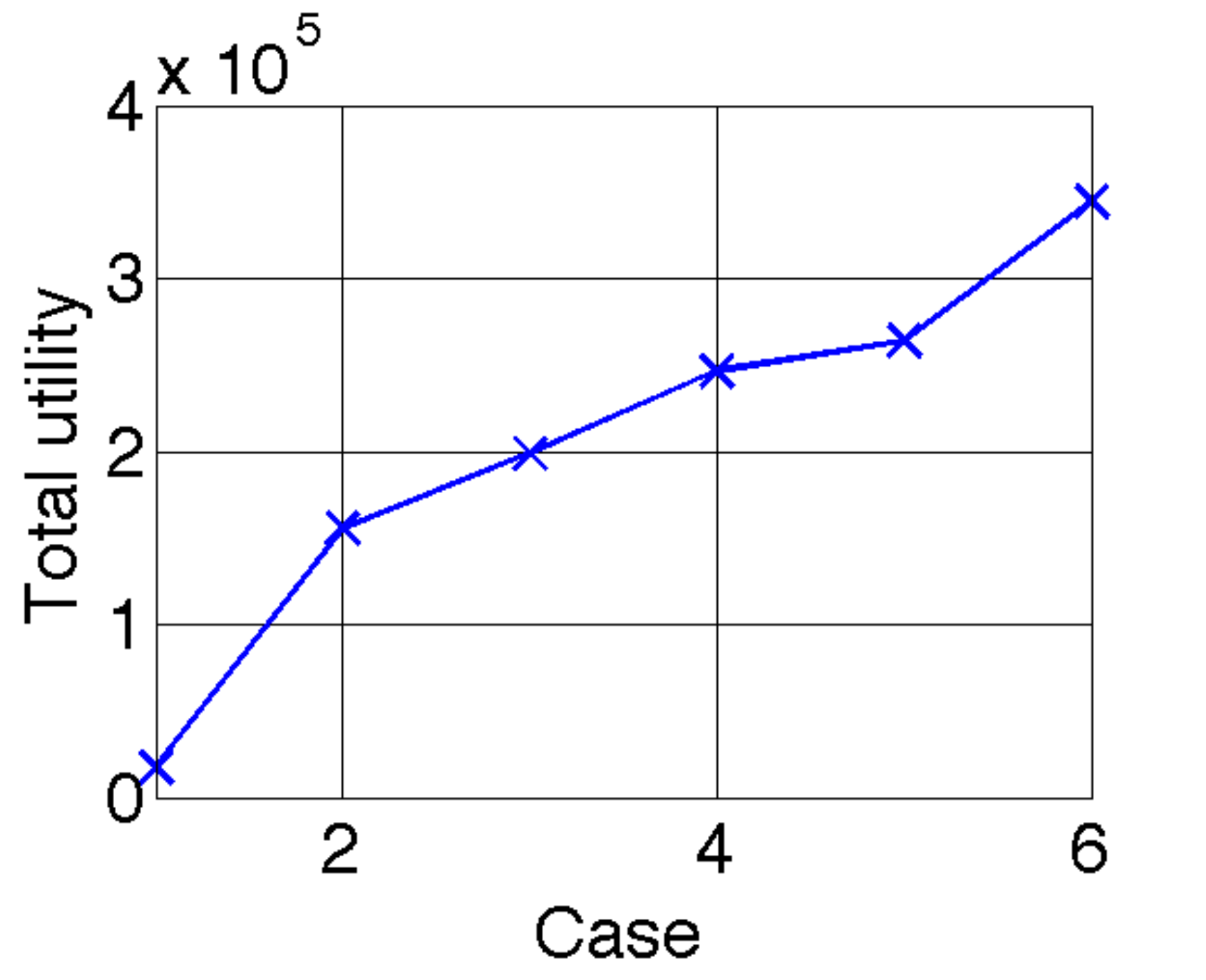}%
 \label{fig:sim_1000_fixedQ0}}
 \hfil
 \centering
 \subfigure[Supply and demand curves]{\includegraphics[width=1.7in]{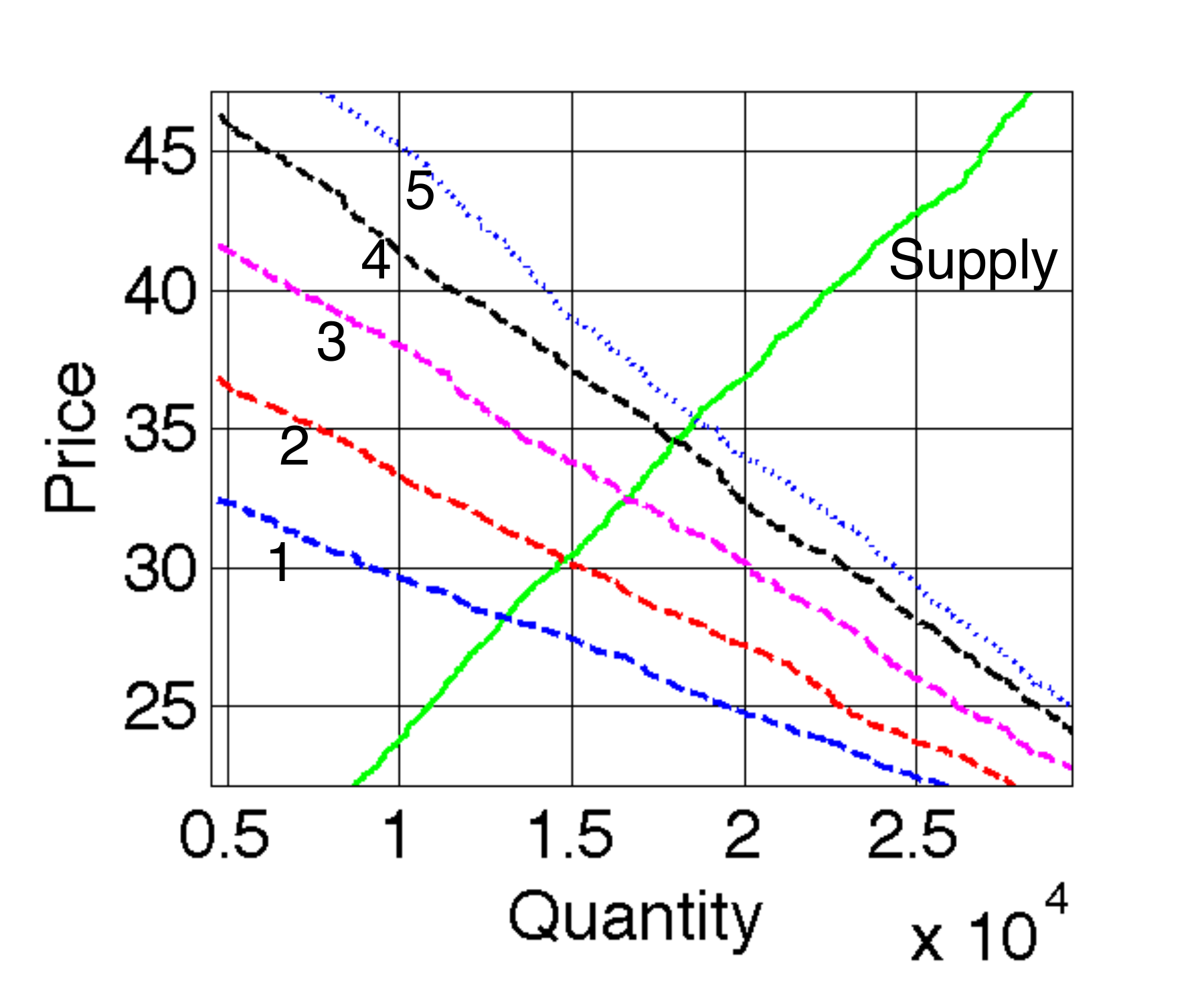}%
 \label{fig:sim_1000_fixedQ0_sd}}
 \caption{Fixed $Q_0$ with 1000 buyers (K) and 1000 aggregators (N)}
 \label{fig:sim1000_fixedQ0}
     \vspace{-.2in}
 \end{figure}
Next we study the results for six different demand curves with a fixed supply curve (i.e. $\alpha$ fixed). Fig.~\ref{fig:sim1000_fixedQ0}\footnotemark[\value{footnote}] shows the six cases with a fixed $Q_0$. The simulation also aligns the results deduced in Section \ref{subsec:linearhomo}: the larger $\beta$, the smaller the utility.

\section{Conclusion}\label{section:conclusion}
We have studied the V2G energy trading market, in which PHEVs sell their excessive energy to some loads in the distribution network. 
We have proposed a multi-layer system to coordinate the trading where those PHEVs in a close geographical area are represented by an aggregator. 
In the macro layer, the energy buyers and the aggregators engage in a double auction to determine the trading price 
and the actual energy quantity sold for each aggregator. In the micro layer, each aggregator helps its managed PHEVs maximize their utilities. 
A mechanism is proposed to coordinate the interaction between the macro and micro layers. 
We have analyzed the system performance when the numbers of buyers and aggregators are large. 
Simulation results show that our mechanism is always better than the greedy approach and verify some of our analytical results.



\bibliographystyle{hieeetr}
\bibliography{mybib}

\end{document}